\documentclass[twoside]{article}

\usepackage[dvips]{graphicx}
\usepackage[utf8]{inputenc}
\usepackage[english,russian]{babel}
\usepackage{amsmath,amsthm,amssymb}
\usepackage{euscript}

\RequirePackage[%
    colorlinks=true,
    pdfstartview=FitH,
    linkcolor=blue,
    citecolor=blue,
    filecolor=blue,
    urlcolor=blue,
    linktoc=page]{hyperref}

\hypersetup{%
    pdftitle={General Upper Bounds for Gate Complexity and Depth of Reversible Circuits Consisting of NOT, CNOT and 2-CNOT Gates},
    pdfauthor={Dmitry V. Zakablukov},
    pdfkeywords={reversible logic, gate complexity, computations with memory},
    pdfsubject={The paper discusses the gate complexity and the depth of reversible circuits consisting of NOT, CNOT and 2-CNOT gates
        in the case of limited number of additional inputs}}

\RequirePackage{array}
\RequirePackage{multirow}

\newcommand{\ZZ}{\mathbb Z}
\newcommand{\frS}{\mathfrak S}
\newcommand{\vv}{\mathbf}

\theoremstyle{plain}
\newtheorem{theorem}{Теорема}
\newtheorem{lemma}{Лемма}

\theoremstyle{definition}

\begin{document}

\title{О зависимости сложности и глубины обратимых схем, состоящих из функциональных элементов NOT, CNOT и 2-CNOT,
    от количества дополнительных входов\\General Upper Bounds for Gate Complexity and Depth of Reversible Circuits %
    Consisting of NOT, CNOT and 2-CNOT Gates}

\author{Д.\,В.~Закаблуков\footnote{Аспирант МГТУ им. Н.\,Э.~Баумана, кафедра <<Информационная безопасность>>, г.~Москва, %
    e-mail: \texttt{\href{mailto:dmitriy.zakablukov@gmail.com}{dmitriy.zakablukov@gmail.com}}}\\
    Dmitry V. Zakablukov\footnote{Post-graduate of Information Security Chair, BMSTU, Moscow, %
    e-mail: \texttt{\href{mailto:dmitriy.zakablukov@gmail.com}{dmitriy.zakablukov@gmail.com}}}}

\maketitle

УДК 004.312, 519.7    
   
Исследование выполнено при финансовой поддержке РФФИ в рамках научного проекта № 16-01-00196~A.

\begin{abstract}
В работе рассматривается вопрос сложности и глубины обратимых схем, состоящих из функциональных элементов NOT, CNOT и 2-CNOT,
в условиях ограничения на количество используемых дополнительных входов.
Изучаются функции Шеннонa сложности $L(n, q)$ и глубины $D(n,q)$ обратимой схемы, реализующей отображение $f\colon \ZZ_2^n \to \ZZ_2^n$,
при условии, что количество дополнительных входов $q$ находится в диапазоне $8n < q \lesssim n2^{n-o(n)}$.
Доказываются верхние оценки $L(n,q) \lesssim 2^n + 8n2^n \mathop / (\log_2 (q-4n) - \log_2 n - 2)$ и
$D(n,q) \lesssim 2^{n+1}(2,5 + \log_2 n - \log_2 (\log_2 (q - 4n) - \log_2 n - 2))$ для указанного диапазона значений $q$.

The paper discusses the gate complexity and the depth of reversible circuits consisting of NOT, CNOT and 2-CNOT gates
in the case, when the number of additional inputs is limited.
We study Shannon's gate complexity function $L(n, q)$ and depth function $D(n, q)$ for a reversible circuit implementing a
Boolean transformation $f\colon \ZZ_2^n \to \ZZ_2^n$ with $8n < q \lesssim n2^{n-o(n)}$ additional inputs.
The general upper bounds $L(n,q) \lesssim 2^n + 8n2^n \mathop / (\log_2 (q-4n) - \log_2 n - 2)$ and
$D(n,q) \lesssim 2^{n+1}(2,5 + \log_2 n - \log_2 (\log_2 (q - 4n) - \log_2 n - 2))$ are proved
for this case.
\end{abstract}

\textbf{Ключевые слова}: обратимые схемы, сложность схемы, глубина схемы, вычисления с памятью.


\section*{Введение}
Теория схемной сложности берет свое начало с работы Шеннона~\cite{shannon}, в которой было предложено в качестве меры сложности
булевой функции рассматривать сложность реализующей ее минимальной контактной схемы.

О.\,Б. Лупановым установлена~\cite{lupanov_complexity} асимптотика сложности $L(n) \sim \rho 2^n \mathop / n$
булевой функции от $n$ переменных в произвольном конечном полном базисе элементов с произвольными положительными весами,
где $\rho$ обозначает минимальный приведенный вес элементов базиса.
Также в работе~\cite{lupanov_delay} им были рассмотрены схемы из функциональных элементов с задержками и было доказано,
что в регулярном базисе функциональных элементов любая булева функция может быть реализована схемой,
имеющей задержку $T(n) \sim \tau n$, где $\tau$~--- минимум приведенных задержек всех элементов базиса, при сохранении
асимптотически оптимальной сложности. Однако вопрос зависимости значений функций $L(n)$ и $T(n)$
от количества используемых регистров памяти в данных работах не рассматривался.

Вопрос о вычислениях с ограниченной памятью был рассмотрен Н.\,А. Карповой в работе~\cite{karpova},
где доказано, что в базисе классических функциональных элементов, реализующих все \mbox{$p$-местные} булевы функции,
асимптотическая оценка функции Шеннона сложности схемы с тремя и более регистрами памяти
зависит от значения $p$, но не изменяется при увеличении количества используемых регистров памяти.

В данной работе рассматриваются схемы, состоящие из обратимых функциональных элементов NOT, CNOT и 2-CNOT.
Определение таких функциональных элементов и схем было дано, например, в работах~\cite{feynman,maslov_thesis,shende}.
Функции Шеннона сложности $L(n,q)$ и глубины $D(n,q)$ обратимой схемы, состоящей из элементов NOT, CNOT и 2-CNOT
и реализующей некоторое булево отображение $\ZZ_2^n \to \ZZ_2^n$ с использованием $q$ дополнительных входов,
были определены в работах~\cite{my_dm_complexity,my_vestnik_mgu_depth}.

В работе~\cite{my_dm_complexity} были получены следующие оценки сложности обратимой схемы:
\begin{gather*}
    L(n,0) \asymp n2^n \mathop / \log_2 n \; , \\
    L(n,q_0) \asymp 2^n \text{ \,при\, } q_0 \sim n 2^{n-\lceil n \mathop / \phi(n)\rceil} \; ,
\end{gather*}
где $\phi(n) \leqslant n \mathop / (\log_2 n + \log_2 \psi(n))$
и $\psi(n)$~--- любые сколь угодно медленно растущие функции.

В работе~\cite{my_vestnik_mgu_depth} были получены следующие оценки глубины обратимой схемы:
\begin{gather*}
    D(n,q) \geqslant \frac{2^n(n-2)}{3(n+q)\log_2(n+q)} - \frac{n}{3(n+q)} \; , \\
    D(n,0) \leqslant \frac{n2^{n+5}}{\psi(n)} \left( 1 + \varepsilon(n) \right) \; ,
\end{gather*}
где $\psi(n) = \log_2 n - \log_2 \log_2 n - \log_2 \phi(n)$,
$\phi(n) < n \mathop / \log_2 n$~--- любая сколь угодно медленно растущая функция,
\begin{equation*}
    \varepsilon(n) = \frac{1}{4\phi(n)} +(4 + o(1))\frac{\log_2 n \cdot \log_2 \log_2 n}{n} \; .
\end{equation*}
Для схем с дополнительными входами были получены следующие оценки глубины обратимой схемы:
\begin{gather}
    D(n,q_1) \lesssim 3n \text{ \,при\, } q_1 \sim 2^n \; ,
        \label{formula_linear_depth_bound}\\
    D(n,q_2) \lesssim 2n \text{ \,при\, } q_2 \sim \phi(n)2^n  \notag \; ,
\end{gather}
где $\phi(n) = o(n)$~--- любая сколь угодно медленно растущая функция.

Таким образом, на сегодняшний день известны оценки сложности и глубины обратимой схемы только в двух крайних случаях:
когда в схеме совсем не используются дополнительные входы и когда их количество весьма велико.
В данной работе при помощи алгоритма синтеза, основанного на стандартном методе О.\,Б.~Лупанова, доказывается,
что для любого значения $q$ в диапазоне $8n < q \lesssim n2^{n-o(n)}$ верны соотношения
$L(n,q) \lesssim 2^n + 8n2^n \mathop / (\log_2 (q-4n) - \log_2 n - 2)$ и
$D(n,q) \lesssim 2^{n+1}(2,5 + \log_2 n - \log_2 (\log_2 (q - 4n) - \log_2 n - 2))$.

\section{Основные понятия}

Определение обратимых функциональных элементов, в частности NOT и $k$-CNOT, можно найти в работах~\cite{feynman,maslov_thesis,shende}.
Mы будем пользоваться формальным определением этих элементов из работы~\cite{my_dm_complexity}.

Напомним, что через $N_j^n$ обозначается функциональный элемент NOT (инвертор) с $n$ входами,
задающий преобразование $\ZZ_2^n \to \ZZ_2^n$ вида
$$
    f_j(\langle x_1, \ldots, x_n \rangle) = \langle x_1, \ldots, x_j \oplus 1, \ldots, x_n \rangle  \; .
$$
Через $C_{i_1,\ldots,i_k;j}^n = C_{I;j}^n$, $j \notin I$, обозначается функциональный элемент $k$-CNOT с $n$ входами
(контролируемый инвертор, обобщенный элемент Тоффоли с $k$ контролирующими входами),
задающий преобразование $\ZZ_2^n \to \ZZ_2^n$ вида
$$
    f_{i_1,\ldots,i_k;j}(\langle x_1, \ldots, x_n \rangle) =
        \langle x_1, \ldots, x_j \oplus x_{i_1} \wedge \ldots \wedge x_{i_k}, \ldots, x_n \rangle  \; .
$$
Будем рассматривать только функциональные элементы NOT, CNOT (1-CNOT) и 2-CNOT и обратимые схемы, состоящие из них.

Через $L(\frS)$, $D(\frS)$ и $Q(\frS)$ будем обозначать сложность обратимой схемы $\frS$, ее глубину и количество дополнительных входов
соответственно.
Функции Шеннона сложности $L(n,q)$ и глубины $D(n,q)$ обратимой схемы, состоящей из элементов NOT, CNOT и 2-CNOT
и реализующей некоторое булево отображение $\ZZ_2^n \to \ZZ_2^n$ с использованием $q$ дополнительных входов,
были определены в работах~\cite{my_dm_complexity,my_vestnik_mgu_depth}.

Значимыми входами схемы будем называть все входы, не являющиеся дополнительными, а значимыми выходами~--- те выходы,
значения на которых нужны для дальнейших вычислений.

\section{Зависимость сложности обратимой схемы от количества дополнительных входов}

Рассмотрим обратимую схему $\frS_n$ на рис.~\ref{pic_lemma_complexity_of_all_conjunctions},
реализующую все конъюнкции $n$ переменных вида
$x_1^{a_1} \wedge \ldots \wedge x_n^{a_n}$, $a_i \in \ZZ_2$. Схема имеет $n$ значимых входов и $2^n$ значимых выходов;
в нее входят подсхемы $\frS_{\lceil n \mathop / 2 \rceil}$ и $\frS_{\lfloor n \mathop / 2 \rfloor}$,
реализующие конъюнкции от меньшего числа переменных.
Если все $2^n$ конъюнкций на значимых выходах основной схемы реализовать одновременно, а не по мере необходимости,
то $L(\frS_n) \sim 2^n$ и $Q(\frS_n) \sim 2^n$,
как было доказано в работе~\cite[Лемма~1]{my_dm_complexity}.
С другой стороны, мы можем конструировать конъюнкции по мере необходимости по одной, а не все сразу, используя только лишь значимые выходы
подсхем $\frS_{\lceil n \mathop / 2 \rceil}$ и $\frS_{\lfloor n \mathop / 2 \rfloor}$ и всего один дополнительный вход,
который и будет хранить значение нужной нам конъюнкции. После того, как все необходимые операции с этим значимым выходом будут
осуществлены, мы можем его обнулить, применив те же функциональные элементы, что и для его получения, но в обратном порядке.
Таким образом, в рассматриваемом нами случае для получения каждой конъюнкции потребуется не более двух элементов 2-CNOT,
а для получения $t$ конъюнкций (последовательно, по мере необходимости)~--- не более $2t$ элементов 2-CNOT.
Следовательно, $L(\frS_n) \lesssim O(2^{n \mathop / 2}) + 2t$, $Q(\frS_n) \lesssim O(2^{n \mathop / 2}) + 1$.
Такой же подход можно применить к подсхемам
$\frS_{\lceil n \mathop / 2 \rceil}$ и $\frS_{\lfloor n \mathop / 2 \rfloor}$, а также к подсхемам этих подсхем.

\begin{figure}
    \centering
    \includegraphics{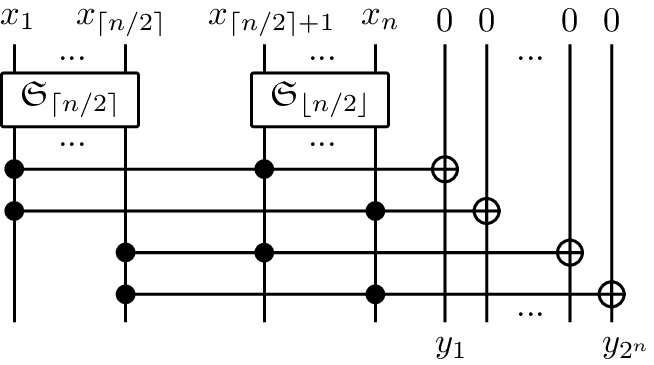}
    \caption
    {
        \small Структура обратимой схемы, реализующей все конъюнкции от $n$ переменных
        с минимальной сложностью (входы схемы сверху).
    }\label{pic_lemma_complexity_of_all_conjunctions}
\end{figure}    

Если вообще не хранить промежуточных значений, а конструировать конъюнкции по мере необходимости, имея лишь входы
$x_1, \ldots, x_n, \bar x_1, \ldots, \bar x_n$, то для получения каждой конъюнкции $x_1^{a_1} \wedge \ldots \wedge x_n^{a_n}$,
очевидно, потребуется не более $2(n-1)$ элементов 2-CNOT, а дополнительных входов потребуется всего $(n-1)$ на все конъюнкции.
К примеру, на рис.~\ref{pic_construct_conjunctions_on_demand} показан пример конструирования 
конъюнкции $\bar x_1 \bar x_2 \bar x_3 x_4 x_5 \bar x_6 x_7 x_8$ с
использованием промежуточных значений $\bar x_1 \bar x_2$, $\bar x_3 x_4$, $x_5 \bar x_6$
и $x_7 x_8$ и последующем обнулением значений на незначимых выходах.

\begin{figure}
    \centering
    \includegraphics{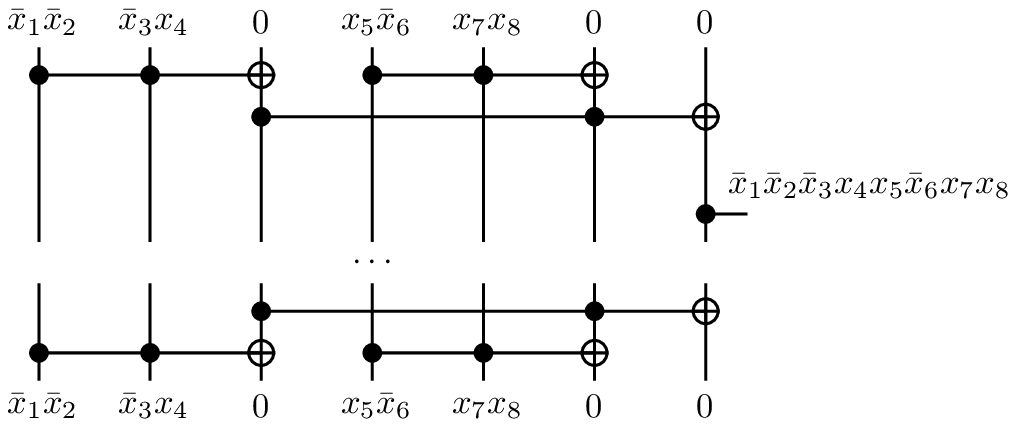}
    \caption
    {
        \small Пример конструирования конъюнкции с использованием промежуточных значений
        и последующим обнулением значений на незначимых выходах (входы схемы сверху).
    }\label{pic_construct_conjunctions_on_demand}
\end{figure}

Рассмотрим в общем случае обратимую схему $\frS_{CONJ(n,q)}$, которая реализует конъюнкции
$n$ переменных вида $x_1^{a_1} \wedge \ldots \wedge x_n^{a_n}$, $a_i \in \ZZ_2$,
при условии, что для хранения промежуточных значений отведено $q$ дополнительных входов,
а значения $\bar x_1, \ldots, \bar x_n$ уже получены ранее.
Обозначим через $L_{CONJ}(n, q, t)$ сложность схемы $\frS_{CONJ(n,q)}$,
реализующей по мере необходимости $t$ конъюнкций, не обязательно различных, причем значение $t$ может быть любым,
в том числе больше $2^n$.
Также обозначим через $Q_{CONJ}(n, q, t)$ общее количество необходимых дополнительных входов для такой обратимой схемы.
Из рассуждений выше можно вывести следующие простые оценки:
\begin{align}
    L_{CONJ}(n, 0, t) & \leqslant 2(n-1)t
        \label{formula_L_CONJ_0}   \; , \\
    Q_{CONJ}(n, 0, t) & = n-1
        \label{formula_Q_CONJ_0}   \; .
\end{align}
Для $q_0 \sim \lceil n \mathop / 2 \rceil + \lfloor n \mathop / 2 \rfloor$ верны соотношения
\begin{gather*}
    L_{CONJ}(n, q_0, t) \leqslant q_0 + 2t  \; , \\
    Q_{CONJ}(n, q_0, t) \leqslant q_0 + 1  \; .
\end{gather*}

Выведем зависимость значения функции $L_{CONJ}(n, q, t)$ от значения $q$.
\begin{lemma}\label{lemma_L_CONJ_bound}
    Для любого значения $q > 2n$, $q \lesssim 2^n$ верны соотношения
    \begin{align}
        L_{CONJ}(n, q, t) & \leqslant q + \frac{8nt}{\log_2 q - \log_2 n - 1}
            \label{formula_L_CONJ_q}   \; , \\
        Q_{CONJ}(n, q, t) & \leqslant q + n - 1
            \label{formula_Q_CONJ_q}   \; .
    \end{align}
\end{lemma}
\begin{proof}
    Соотношение $Q_{CONJ}(n, q, t) \leqslant q + n - 1$ следует из соотношения~\eqref{formula_Q_CONJ_0}.

    \begin{figure}[ht]
        \centering
        \includegraphics[scale=0.95]{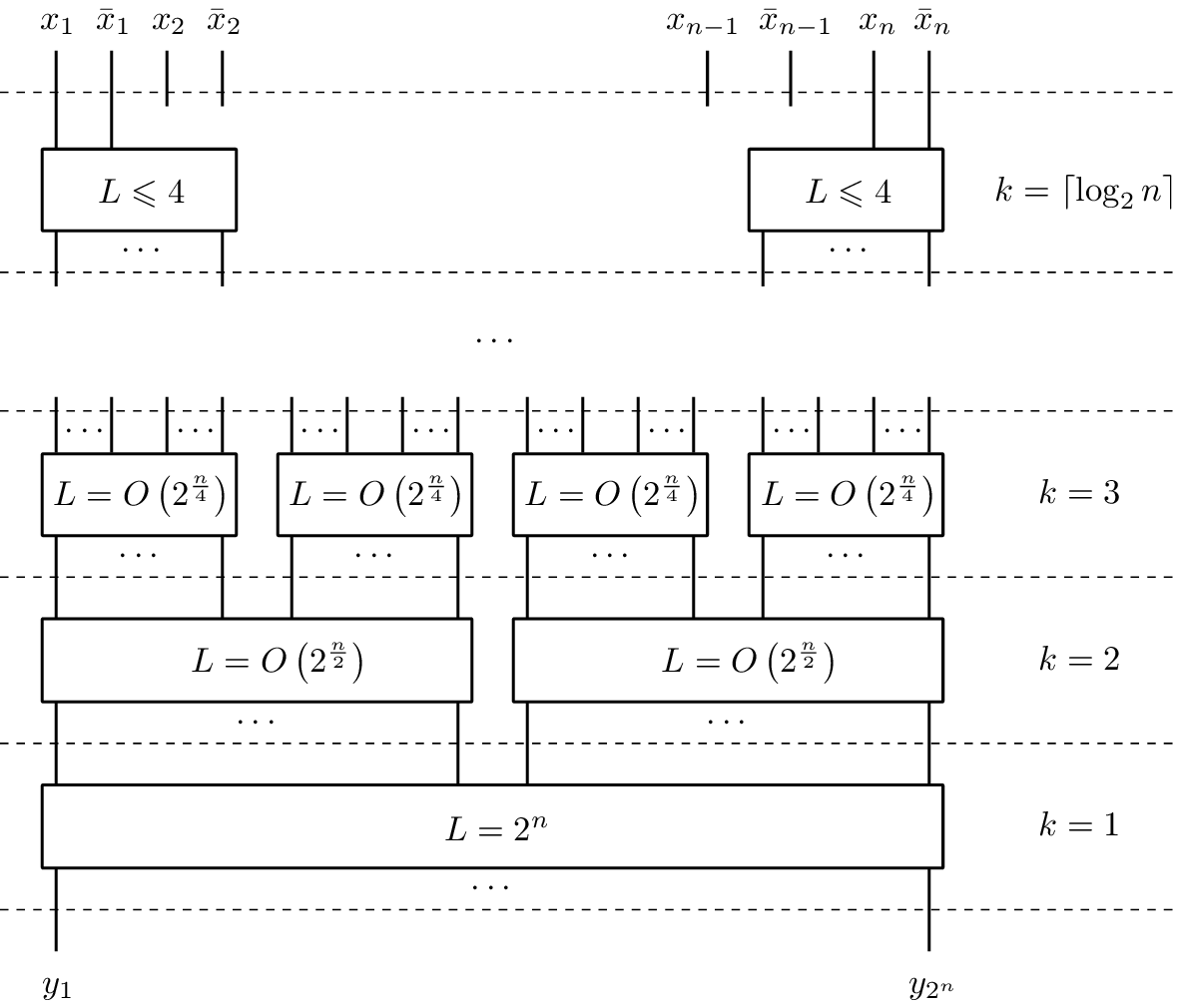}
        \caption
        {
            \small Общая структура обратимой схемы $\frS_{CONJ(n,q)}$ (входы схемы сверху).
        }\label{pic_S_CONJ_structure}
    \end{figure}

    Рассмотрим структуру искомой обратимой схемы $\frS_{CONJ(n,q)}$ на рис.~\ref{pic_S_CONJ_structure}:
    она разбита на $K = \lceil \log_2 n \rceil$ уровней, нумерация ведется снизу вверх.
    На уровне номер $k$ расположены $2^{k-1}$ обратимых подсхем, все они имеют примерно одинаковое количество значимых входов и выходов
    и реализуют все конъюнкции от некоторого подмножества переменных $x_1, \ldots, x_n$, причем подмножества для разных схем
    одного уровня не пересекаются, их объединение равно всему множеству $\{\,x_1, \ldots, x_n\,\}$, а мощности
    данных подмножеств примерно равны.

    Для пояснения структуры схемы $\frS_{CONJ}$ рассмотрим частный ее случай для $n = 7$.
    Схема имеет $K = 3$ уровня, каждый ее уровень расписан в Таблице~\ref{table_S_CONJ_example}.
    Из данной таблицы видно, что если некоторая подсхема $\frS_{k;i}$ на уровне $k$ имеет $2^m$ значимых выходов,
    то на уровне $(k+1)$ есть ровно две подсхемы $\frS_{k+1;j}$ и $\frS_{k+1;j+1}$, подключенные к ней,
    первая из которых имеет $2^{\lfloor m \mathop / 2 \rfloor}$ значимых выходов, а вторая~---
    $2^{\lceil m \mathop / 2 \rceil}$ значимых выходов.
    Структура подсхемы $\frS_{k;i}$ проста: она реализует конъюнкции каждого значимого выхода подсхемы $\frS_{k+1;j}$ с каждым значимым
    выходом подсхемы $\frS_{k+1;j+1}$ (см. рис.~\ref{pic_sub_S_CONJ_structure}). Следовательно, сложность такой подсхемы будет равна
    $2^{\lfloor m \mathop / 2 \rfloor} \cdot 2^{\lceil m \mathop / 2 \rceil} = 2^m$ (используются только элементы 2-CNOT).
    
    \begin{table}[ht]
        \centering
        \begin{tabular}{|m{0.6cm}|m{1.7 cm}|m{1.5 cm}|m{1.7 cm}|m{4 cm}|}
            \hline
            \centering \textbf{Ур.} &
            \centering \textbf{Подсхема} &
            \centering \textbf{Кол-во входов} &
            \centering \textbf{Кол-во выходов} &
            \centering \textbf{Подмножество переменных, для которых реализованы все конъюнкции} \tabularnewline
            \hline

            \centering 1 &
            \centering $\frS_{1;1}$ &
            \centering $2^3 + 2^4$ &
            \centering $2^7$ &
            \centering $\{\,x_1, \ldots, x_7\,\}$ \tabularnewline
            \hline

            \multirow{2}{0.6cm}{\centering 2} &
            \centering $\frS_{2;1}$ &
            \centering $2^1 + 2^2$ &
            \centering $2^3$ &
            \centering $\{\,x_1, x_2, x_3\,\}$ \tabularnewline
            \cline{2-5}

            &
            \centering $\frS_{2;2}$ &
            \centering $2^2 + 2^2$ &
            \centering $2^4$ &
            \centering $\{\,x_4, x_5, x_6, x_7\,\}$ \tabularnewline
            \hline

            \multirow{4}{0.6cm}{\centering 3} &
            \centering $\frS_{3;1}$ &
            \centering 2 &
            \centering $2^1$ &
            \centering $\{\,x_1\,\}$ \tabularnewline
            \cline{2-5}

            &
            \centering $\frS_{3;2}$ &
            \centering 4 &
            \centering $2^2$ &
            \centering $\{\,x_2, x_3\,\}$ \tabularnewline
            \cline{2-5}

            &
            \centering $\frS_{3;3}$ &
            \centering 4 &
            \centering $2^2$ &
            \centering $\{\,x_4, x_5\,\}$ \tabularnewline
            \cline{2-5}

            &
            \centering $\frS_{3;4}$ &
            \centering 4 &
            \centering $2^2$ &
            \centering $\{\,x_6, x_7\,\}$ \tabularnewline
            \hline
        \end{tabular}
        \caption
        {
            Описание структуры обратимой схемы $\frS_{CONJ}$ при $n=7$.
        }\label{table_S_CONJ_example}
    \end{table}
        
    \begin{figure}[ht]
        \centering
        \includegraphics[scale=0.95]{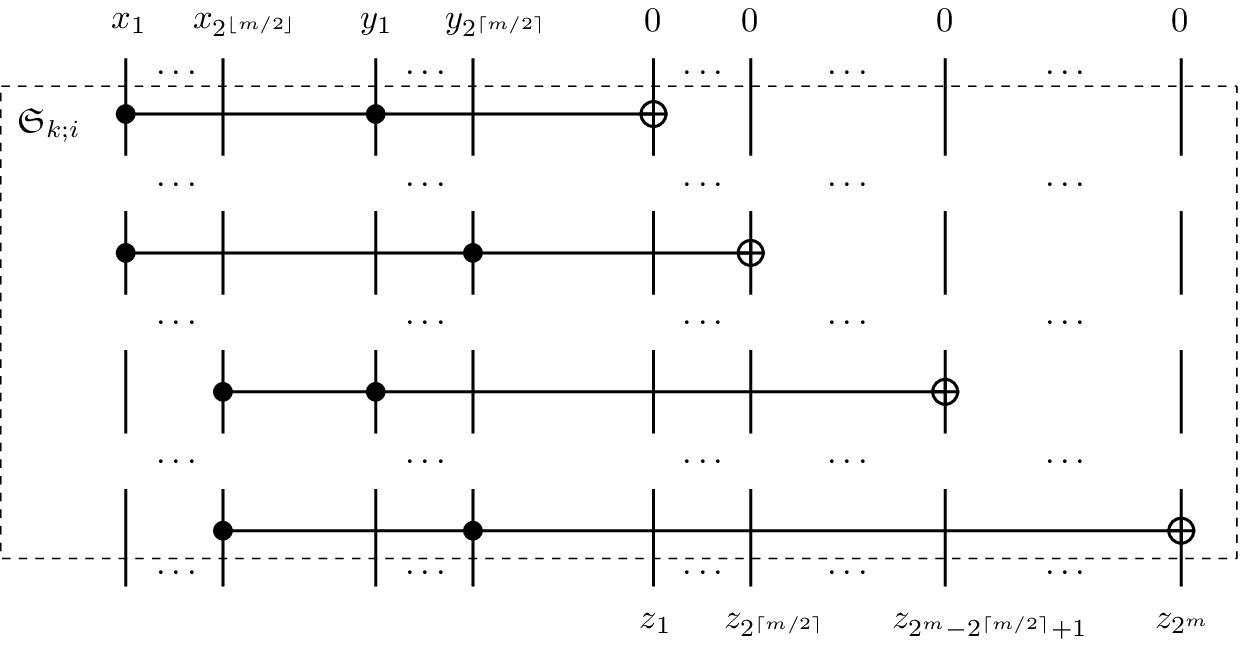}
        \caption
        {
            \small Структура подсхемы $\frS_{k;i}$ обратимой схемы $\frS_{CONJ}$ (входы схемы сверху).
        }\label{pic_sub_S_CONJ_structure}
    \end{figure}

    Вернемся к общей схеме $\frS_{CONJ}$. Нам дано $q$ дополнительных входов для хранения промежуточных значений.
    Разумнее всего потратить их для хранения значений на выходах подсхем самых высоких уровней, поскольку видно, что чем меньше
    уровень схемы $\frS_{CONJ}$, тем больше требуется дополнительных входов для хранения промежуточных значений.
    
    Рассмотрим случай, когда мы имеем возможность хранить все промежуточные значения.
    Обозначим через $L_k$ количество элементов на $k$-м уровне схемы. К примеру, $L_1 = 2^n$,
    $L_2 = 2^{\lfloor n \mathop / 2 \rfloor} + 2^{\lceil n \mathop / 2 \rceil}$.

    Оценим значение $L_k$. Поскольку
    $$
        \left\lceil \frac{n}{2} \right\rceil = \left\lfloor \frac{n + 1}{2} \right\rfloor \leqslant \frac{n + 1}{2}  \; ,
    $$
    то
    $$
        L_2 \leqslant 2 \cdot \max\left(2^{\left\lfloor \frac{n}{2} \right\rfloor}, 2^{\left\lceil \frac{n}{2} \right\rceil} \right) =
            2 \cdot 2^{\left\lceil \frac{n}{2} \right\rceil} \leqslant 2 \cdot 2 ^ {\frac{n}{2} + \frac{1}{2}}  \; .
    $$
    Отсюда следует, что
    \begin{gather*}
        L_3 \leqslant 4 \cdot 2 ^ {\frac{n}{4} + \frac{1}{4} + \frac{1}{2}}  \; ,  \\
        L_4 \leqslant 8 \cdot 2 ^ {\frac{n}{8} + \frac{1}{8} + \frac{1}{4} + \frac{1}{2}}  \; ,  \\
        L_k \leqslant 2^k \cdot 2 ^ {n \mathop / 2^{k-1}}  \; .
    \end{gather*}
    Обозначим $\delta_k = 2^k \cdot 2 ^ {n \mathop / 2^{k-1}}$. Значение переменной $k$ лежит в диапазоне
    $[1, \ldots, K]$, $K = \lceil \log_2 n \rceil$, $k \in \mathbb N$.
    Сделаем переобозначение переменной: $k = K - s$, тогда $s = s(k) = K - k$. Если $k$ обозначает номер уровня схемы
    при нумерации от выходов ко входам (снизу вверх), то $(s+1)$ будет означать номер уровня схемы при нумерации от входов к выходам
    (сверху вниз).
    Значение переменной $s$ лежит в диапазоне $[0, \ldots, K - 1]$, $s \in \mathbb Z_+$.
    В этом случае
    $$
        \delta_k = \frac{2^K}{2^s} \cdot 2^{(2n \cdot 2^s) \mathop / 2^K}
            \leqslant \frac{2n}{2^s} \cdot 2^{2^{s+1}} = \Delta_s  \; .
    $$
    Следовательно, мы получили цепочку неравенств
    $$
        L_k \leqslant \delta_k \leqslant \Delta_{s(k)} = \frac{2n}{2^s} \cdot 2^{2^{s+1}}  \; .
    $$
    
    Выпишем первые члены ряда $\{\,\Delta_{s(k)}\,\}$: $\{\,8n, 16n, 128n, \ldots\,\}$. Видно, что с ростом $s$ значение
    $\Delta_s$ растет все быстрее. Более того, можно утверждать, что для любого $s \geqslant 1$ верно соотношение
    $$
        \sum_{i=0}^{s-1} {\Delta_i} \leqslant \frac{\Delta_s}{2}  \; .
    $$
    Отсюда следует, что
    \begin{gather}
        \sum_{i=0}^s {\Delta_i} \leqslant \frac{3\Delta_s}{2}  \; , \notag \\
        \sum_{i=K}^{K-s} L_i \leqslant \frac{3n}{2^s} \cdot 2^{2^{s+1}}  \; .
            \label{formula_complexity_of_last_layers}
    \end{gather}
    Другими словами, сумма сложностей всех подсхем на последних $(s+1)$ уровнях (при нумерации снизу вверх) не превышает
    $3n \cdot 2^{-s + 2^{s+1}}$.
    
    Вернемся снова к общей схеме $\frS_{CONJ}$.
    Из рис.~\ref{pic_construct_conjunctions_on_demand} видно,
    что для конструирования по мере необходимости одного значимого выхода схемы $\frS_{CONJ}$ на первых $r$ уровнях
    будет использовано не более $(1 + 2 + 4 + \ldots + 2^{r-1})=(2^r - 1)$ элементов 2-CNOT.
    Столько же функциональных элементов потребуется для обнуления
    значений на незначимых выходах. Следовательно, при условии, что количество уровней, для которых подсхемы
    надо конструировать по мере необходимости, не превышает $r$, верно соотношение
    \begin{equation}
        L_{CONJ}(n, q, t) \leqslant q + 2(2^r - 1)\cdot t \leqslant q + t \cdot 2^{r+1}  \; .
            \label{formula_bound_for_L_CONJ_with_r}
    \end{equation}
    Слагаемое $q$ в данном соотношении, очевидно, следует из того факта, что для получения значения на одном выходе
    любой подсхемы $\frS_{k;i}$ требуется ровно один элемент 2-CNOT.
    Если мы можем хранить не более $q$ промежуточных значений на выходах подсхем,
    то для их хранения потребуется не более $q$ элементов 2-CNOT.

    Нам требуется оценить значение $r$. Пусть для данного по условию задачи значения $q$ выполняется неравенство
    \begin{equation}
        \frac{3n}{2^s} \cdot 2^{2^{s+1}} \leqslant q < \frac{3n}{2^{s+1}} \cdot 2^{2^{s+2}}  \;
            \label{formula_bounds_for_q_in_S_CONJ}
    \end{equation}
    для некоторого значения $s \in [0, \ldots, K - 1]$, $s \in \mathbb Z_+$. Тогда мы можем утверждать,
    что данного значения $q$ достаточно для хранения значений на всех выходах подсхем на последних $(s+1)$ уровнях
    при нумерации уровней снизу вверх (см. соотношение~\eqref{formula_complexity_of_last_layers}),
    а количество первых уровней, для которых подсхемы нужно конструировать по мере необходимости,
    не превышает $(k-1)$, поскольку $(s+1) = K - (k-1)$.
    Следовательно, $r \leqslant k-1$.
    Из правого неравенства соотношения~\eqref{formula_bounds_for_q_in_S_CONJ} следует, что
    \begin{gather*}
        \log_2 q < \log_2 3 + \log_2 n - s - 1 + 2^{s+2}  \; , \\
        \frac{2^K}{2^k} = 2^s > \frac{\log_2 q - (\log_2 3 + \log_2 n - s - 1)}{4}  \; , \\
        (\log_2 q - (\log_2 3 + \log_2 n - s - 1))2^k < 4 \cdot 2^K  \; , \\
        K = \lceil \log_2 n \rceil \;,\;\, s \geqslant 0 \; \Rightarrow
            2^k < \frac{8n}{\log_2 q - \log_2 n - 1}  \; \text{ при $q > 2n$ }  \; .
    \end{gather*}
    
    Поскольку $r + 1 \leqslant k$, то
    \begin{equation}
        2^{r+1} < \frac{8n}{\log_2 q - \log_2 n - 1}  \; \text{ при $q > 2n$ }  \; .
        \label{formula_for_2_pow_r_in_lemma_L_CONJ_bound}
    \end{equation}
    Из этого неравенства и неравенства~\eqref{formula_bound_for_L_CONJ_with_r}
    следует оценка из утверждения Леммы
    $$
        L_{CONJ}(n, q, t) \leqslant q + \frac{8nt}{\log_2 q - \log_2 n - 1}  \; .
    $$
    Ограничение $q \lesssim 2^n$ из утверждения Леммы связано с тем,
    что при значениях $q \gtrsim 2^n$ не наблюдается снижение сложности обратимой схемы для рассматриваемого способа синтеза.
\end{proof}

Отметим, что по аналогии со схемой $\frS_{CONJ}$ можно построить схему $\frS_{XOR}$, которая для заданных входов
$x_1, \ldots, x_n$ реализует на своих значимых выходах значения $x_1 \wedge a_1 \oplus \ldots \oplus x_n \wedge a_n$,
$a_i \in \ZZ_2$. Для этого просто надо каждый элемент 2-CNOT в схеме $\frS_{CONJ}$ заменить на два элемента CNOT.
Следовательно,
\begin{gather*}
    L_{XOR}(n, q, t) \leqslant 2q + \frac{16nt}{\log_2 q - \log_2 n - 1}  \; , \\
    Q_{XOR}(n, q, t) \leqslant q + n - 1  \; .
\end{gather*}

Теперь мы можем доказать основную теорему данного раздела.
\begin{theorem}[общая верхняя оценка сложности обратимой схемы с дополнительными входами]\label{theorem_L_n_q_bound_for_arbitrary_q}
    Для любого значения $q > 8n$, $q \lesssim n 2^{n-\lceil n \mathop / \phi(n)\rceil}$,
    где $\phi(n) \leqslant n \mathop / (\log_2 n + \log_2 \psi(n))$
    и $\psi(n)$~--- любые сколь угодно медленно растущие функции, верно соотношение
    $$
        L(n,q) \lesssim 2^n + \frac{8n2^n}{\log_2 (q-4n) - \log_2 n - 2} \; .
    $$
\end{theorem}
\begin{proof}
    Опишем алгоритм синтеза $\mathbf {A_q}$,
    который является модификацией стандартного алгоритма О.\,Б.~Лупанова и который предназначен для синтезирования обратимых схем
    в условиях ограничения на количество используемых дополнительных входов.
    
    Произвольное булево отображение $f\colon \ZZ_2^n \to \ZZ_2^n$ можно представить в виде некоторых $n$ булевых функций
    $f_i\colon \ZZ_2^n \to \ZZ_2$ от $n$ переменных
    \begin{equation}\label{formula_function_decomposition_by_n_functions}
        f(\vv x) = \langle f_1(\vv x), f_2(\vv x), \ldots, f_n(\vv x) \rangle \; .
    \end{equation}
    Каждую функцию $f_i(\vv x)$ можно разложить по последним $(n-k)$ переменным:
    \begin{equation}\label{formula_function_decomposition_by_last_variables}
        f_i(\vv x) = \bigoplus_{a_{k+1}, \ldots, a_n \in \ZZ_2} {x_{k+1}^{a_{k+1}} \wedge \ldots \wedge x_n^{a_n}}
            \wedge f_i(\langle x_1, \ldots, x_k, a_{k+1}, \ldots, a_n \rangle) \; .
    \end{equation}

    Каждая из $n2^{n-k}$ булевых функций $f_i(\langle x_1, \ldots, x_k, a_{k+1}, \ldots, a_n \rangle)$, $1 \leqslant i \leqslant n$,
    является функцией от $k$ переменных $x_1, \ldots, x_k$, ее можно получить при помощи аналога
    СДНФ, в котором дизъюнкции заменяются на сложение по модулю два:
    \begin{equation}\label{formula_analog_sdnf}
        f_i(\langle x_1, \ldots, x_k, a_{k+1}, \ldots, a_n \rangle) = f_{i,j} = \bigoplus_{
            \substack{\boldsymbol \sigma \in \ZZ_2^k \\f_{i,j}(\boldsymbol \sigma) = 1}}
            x_1^{\sigma_1} \wedge \ldots \wedge x_k^{\sigma_k} \; .        
    \end{equation}
    
    Все $2^k$ конъюнкций вида $x_1^{\sigma_1} \wedge \ldots \wedge x_k^{\sigma_k}$ можно разделить на группы,
    в каждой из которых будет не более $s$ конъюнкций. Обозначим через $p = \lceil 2^k \mathop / s \rceil$ количество таких групп.
    Используя конъюнкции одной группы, мы можем реализовать не более $2^s$ булевых функций по формуле~\eqref{formula_analog_sdnf}.
    Обозначим через $G_i$ множество булевых функций, которые могут быть реализованы при помощи конъюнкций $i$-й группы,
    $1 \leqslant i \leqslant p$. Тогда $|G_i| \leqslant 2^s$.
    Следовательно, мы можем переписать формулу~\eqref{formula_analog_sdnf} следующим образом:
    \begin{equation}\label{formula_analog_sdnf_improved}
        f_i(\langle x_1, \ldots, x_k, a_{k+1}, \ldots, a_n \rangle) = \bigoplus_{
            \substack{t=1 \ldots p\\ g_{j_t} \in G_t\\ 1 \leqslant j_t \leqslant |G_t|}} g_{j_t}(\langle x_1, \ldots, x_k\rangle) \; .
    \end{equation}
    
    Отсюда следует, что
    \begin{equation}\label{formula_f_i_with_braces}
        f_i(\vv x) = \bigoplus_{a_{k+1}, \ldots, a_n \in \ZZ_2}
            \left( \bigoplus_{ \substack{t=1 \ldots p\\ g_{j_t} \in G_t\\ 1 \leqslant j_t \leqslant |G_t|}}
            x_{k+1}^{a_{k+1}} \wedge \ldots \wedge x_n^{a_n} \wedge g_{j_t}(\langle x_1, \ldots, x_k\rangle) \right) \; .        
    \end{equation}

    Общая структура обратимой схемы $\frS_f$, которая реализует отображение $f$
    и которая синтезируется алгоритмом $\mathbf {A_q}$,
    показана на рис.~\ref{pic_scheme_structure_for_case_of_limited_memory}.
    
    \begin{figure}[!ht]
        \includegraphics[scale=0.6]{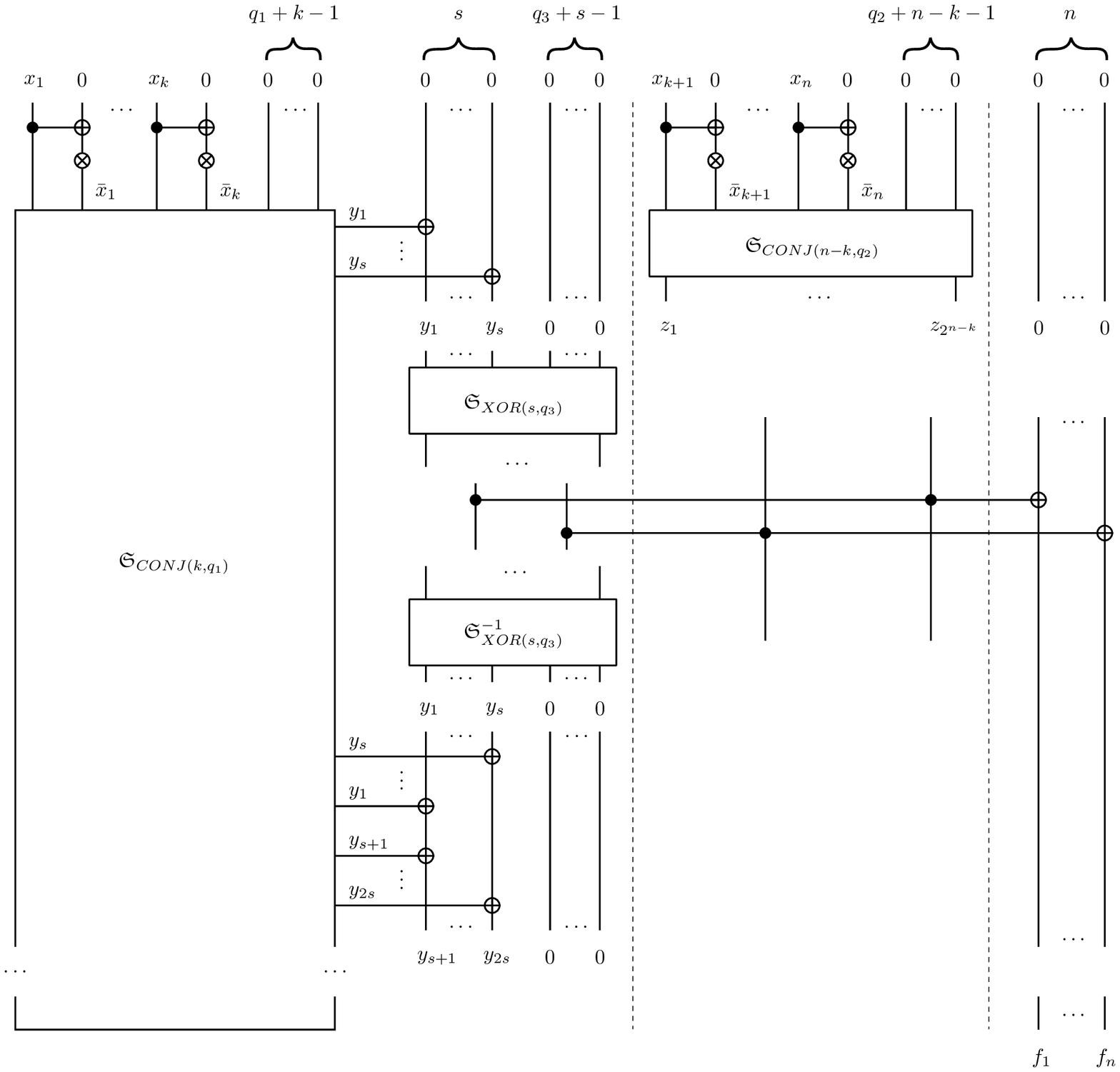}
        \caption
        {
            \small Структура обратимой схемы $\frS_f$, реализующей отображение $f\colon \ZZ_2^n \to \ZZ_2^n$
            в условиях ограничения на количество используемых дополнительных входов (входы схемы сверху).
        }\label{pic_scheme_structure_for_case_of_limited_memory}
    \end{figure}    
    
    Сперва реализуем отрицания для всех входных значений $x_1, \ldots, x_n$ со сложностью $2n$ (по элементу NOT и CNOT на каждый вход),
    задействовав $n$ дополнительных входов.
    
    Разобьем множество значимых входов схемы $x_1, \ldots, x_n$ на две группы: $\{\,x_1, \ldots, x_k\,\}$
    и $\{\,x_{k+1}, \ldots, x_n\,\}$.
    Первую группу входов вместе с их отрицаниями подадим на подсхему $\frS_1 = \frS_{CONJ(k,q_1)}$ для реализации некоторых $t_1$
    конъюнкций $x_1^{a_1} \wedge \ldots \wedge x_k^{a_k}$, $a_i \in \ZZ_2$, отведя данной подсхеме $q_1$ дополнительных входов
    для хранения промежуточных значений.
    Вторую группу входов вместе с их отрицаниями подадим на подсхему $\frS_2 = \frS_{CONJ(n-k,q_2)}$ для реализации некоторых $t_2$
    конъюнкций $x_{k+1}^{a_{k+1}} \wedge \ldots \wedge x_n^{a_n}$, $a_i \in \ZZ_2$, отведя данной подсхеме $q_2$ дополнительных входов
    для хранения промежуточных значений.
    
    Будем реализовывать все $2^k$ различных конъюнкций на значимых выходах подсхемы $\frS_1$ последовательно.
    Полученные значения будем хранить, используя дополнительные входы.
    Как только будут получены очередные $s$ конъюнкций, соответствующие им $s$ значимых выходов
    подаем на значимые входы подсхемы $\frS_{3;i} = \frS_{XOR(s,q_3)}$ для получения значений некоторых $t_3$ функций от переменных
    $x_1, \ldots, x_k$, отведя данной подсхеме $q_3$ дополнительных входов для хранения промежуточных значений.
    Всего будет не более $p = \lceil 2^k \mathop / s \rceil$ различных подсхем $\frS_{3;i}$.
    Как только работа с очередной подсхемой $\frS_{3;i}$ будет закончена, значение на $q_3$ незначимых выходах обнуляем,
    применяя те же самые функциональные элементы, что и для получения подсхемы, но в обратном порядке.
    Затем обнуляем значения на $s$ значимых выходах,
    служивших значимыми входами подсхеме $\frS_{3;i}$, реализуя еще раз полученные ранее $s$ конъюнкций при помощи подсхемы $\frS_1$
    (см. рис.~\ref{pic_scheme_structure_for_case_of_limited_memory}).
    Тем самым мы сможем не увеличивать количество используемых дополнительных входов, а использовать одни и те же дополнительные входы,
    увеличивая однако при этом сложность соответствующих подсхем в два раза.
    
    Из формулы~\eqref{formula_f_i_with_braces} следует,
    что имея значения некоторого значимого выхода подсхемы $\frS_2$ и некоторого значимого выхода подсхемы $\frS_{3;i}$,
    мы можем реализовать одно слагаемое во внутренней скобке, используя ровно один элемент 2-CNOT, контролируемый выход которого
    будет одним из $n$ значимых выходов нашей конструируемой схемы $\frS_f$
    (см. рис.~\ref{pic_scheme_structure_for_case_of_limited_memory}).
    Рассматриваемое нами отображение $f$ имеет $n$ выходов, количество групп конъюнкций от первых $k$ переменных $x_1, \ldots, x_k$
    равно $p$, количество различных конъюнкций от последних $(n-k)$ переменных $x_{k+1}, \ldots, x_n$ равно $2^{n-k}$.
    Следовательно, схемная сложность реализации функции $f_i$ по формуле~\eqref{formula_f_i_with_braces}
    равна $p2^{n-k}$, а отображения $f$ в целом равна $L_4 = pn2^{n-k}$, при этом потребуется ровно $n$
    дополнительных входов для хранения выходных значений отображения $f$.
    
    Таким образом, мы можем вывести неравенство для $L(f,q)$ следующего вида:
    \begin{multline}
        L(f,q) = 2n + L_{CONJ}(k, q_1, t_1) + L_{CONJ}(n-k, q_2, t_2) + \\
            + 2p \cdot L_{XOR}(s, q_3, t_3) + pn2^{n-k}  \; ,
        \label{formula_L_f_q_common}
    \end{multline}
    и для $Q(\frS_f)$ следующего вида:
    \begin{multline}
        Q(\frS_f) = q = n + Q_{CONJ}(k, q_1, t_1) + \\
            + Q_{CONJ}(n-k, q_2, t_2) + Q_{XOR}(s, q_3, t_3) + n  \; .
        \label{formula_Q_f_q_common}
    \end{multline}
    
    Отметим, что каждая из $2^k$ различных конъюнкций на значимых выходах подсхемы $\frS_1$ будет получена ровно два раза,
    следовательно, $t_1 = 2^{k+1}$.
    
    Поскольку каждый значимый выход подсхемы $\frS_2$ используется в качестве входа для $pn$ элементов 2-CNOT,
    а значимый выход подсхемы $\frS_{3;i}$ может использоваться в качестве входа для $2^{n-k}$ элементов 2-CNOT,
    возникает два различных способа конструирования нашей искомой схемы $\frS_f$:
    \begin{enumerate}
        \item \label{item_first_case_for_theorem_L_n_q_bound_for_arbitrary_q}
            В первом случае мы минимизируем значение $t_2$: для каждой группы конъюнкций от первых $k$ переменных $x_1, \ldots, x_k$
            мы один раз конструируем очередной значимый выход подсхемы $\frS_2$, а затем конструируем для него
            $n$ значимых выходов подсхемы $\frS_{3;i}$. Тогда можно утверждать, что $t_2 = p2^{n-k}$, $t_3 \leqslant n2^{n-k}$.
        \item
            Во втором случае мы минимизируем значение $t_3$: для каждой группы конъюнкций от первых $k$ переменных $x_1, \ldots, x_k$
            мы один раз конструируем очередной значимый выход подсхемы $\frS_{3;i}$, а затем конструируем для него нужные значимые выходы
            подсхемы $\frS_2$. Таких выходов может быть один, а может быть и $2^{n-k}$. Однако мы точно можем утверждать,
            что $t_2 \leqslant pn2^{n-k}$, $t_3 \leqslant 2^s$.
    \end{enumerate}
    
    Оценим в общем случае значение $L(f,q)$:
    \begin{multline*}
        L(f,q) \leqslant 2n + pn2^{n-k} + q_1 + q_2 + 4p q_3 +
             \frac{8k2^{k+1}}{\log_2 q_1 - \log_2 k - 1} + \\
                + \frac{8(n-k)t_2}{\log_2 q_2 - \log_2 (n-k) - 1}
                + \frac{32pst_3}{\log_2 q_3 - \log_2 s - 1} \; .
    \end{multline*}
    
    Будем искать такие значения $k$ и $s$, что $p = \lceil 2^k \mathop / s \rceil \sim 2^k \mathop / s$. Тогда
    \begin{multline*}
        L(f,q) \lesssim 2n + \frac{n2^n}{s} + q_1 + q_2 + \frac{4q_3 2^k}{s} +
             \frac{8k2^{k+1}}{\log_2 q_1 - \log_2 k - 1} + \\
                + \frac{8(n-k)t_2}{\log_2 q_2 - \log_2 (n-k) - 1}
                + \frac{32t_3 2^k}{\log_2 q_3 - \log_2 s - 1} \; .
    \end{multline*}
   
    \begin{enumerate}
        \item
            Пусть $t_2 = p2^{n-k} \sim 2^n \mathop / s$, $t_3 \leqslant n2^{n-k}$. В этом случае 
            \begin{multline}
                L(f,q) \lesssim 2n + \frac{n2^n}{s} + q_1 + q_2 + \frac{4q_3 2^k}{s} +
                     \frac{8k2^{k+1}}{\log_2 q_1 - \log_2 k - 1} + \\
                        + \frac{8 \cdot 2^n}{\log_2 q_2 - \log_2 (n-k) - 1}
                        + \frac{32n2^n}{\log_2 q_3 - \log_2 s - 1} \; .
                    \label{formula_L_f_q_first_way}
            \end{multline}
            Положим $s = n - k$, $k = \lceil n \mathop / \phi(n) \rceil$,
            где $\phi(n) \leqslant n \mathop / (\log_2 n + \log_2 \psi(n))$
            и $\psi(n)$~--- любые сколь угодно медленно растущие функции.
            В этом случае будет верно неравенство $2^k \mathop / s \geqslant \psi(n)$.

            Поскольку $q_3 \lesssim 2^s$ и $q_2 < q \lesssim n 2^{n-\lceil n \mathop / \phi(n)\rceil}$,
            то $q_2 = o(2^n)$ и верно соотношение
            \begin{equation}
                2n + \frac{n2^n}{s} + q_2 + \frac{4q_3 2^k}{s} \lesssim 2n + \frac{n2^n}{n - o(n)} + q_2 + \frac{2^{n+2}}{n-o(n)}
                    \lesssim 2^n  \; .
                \label{formula_least_member_bound_first_case_in_theorem_L_n_q_bound_for_arbitrary_q}
            \end{equation}

            Положим $q_1 = 0$, $q_2 = q_3$.    
            Согласно формуле~\eqref{formula_L_CONJ_0}, $L_{CONJ}(n, 0, t) \leqslant 2(n-1)t$,
            следовательно, мы можем заменить в соотношении~\eqref{formula_L_f_q_first_way}
            сложность подсхемы $\frS_1$ на $k2^{k+2}$:
            $$
                L(f,q) \lesssim 2^n + k2^{k+2} + \frac{8 \cdot 2^n}{\log_2 q_2 - \log_2 (n-k) - 1}
                    + \frac{32n2^n}{\log_2 q_3 - \log_2 s - 1} \; .
            $$
            Очевидно, что $k2^{k+2} = 4\lceil n \mathop / \phi(n) \rceil \cdot 2^{\lceil n \mathop / \phi(n) \rceil} = o(2^n)$
            и $8 \cdot 2^n \mathop / (\log_2 q_3 - \log_2 s - 1) = o \left(32n2^n \mathop / (\log_2 q_3 - \log_2 s - 1) \right)$,
            поэтому верно соотношение 
            $$
                L(f,q) \lesssim 2^n + \frac{32n2^n}{\log_2 q_3 - \log_2 s - 1} \; .
            $$

            Согласно Лемме~\ref{lemma_L_CONJ_bound}, $Q_{CONJ}(n, q, t) \leqslant q + n - 1$,
            поэтому соотношение~\eqref{formula_Q_f_q_common} можо переписать в виде
            \begin{equation}
                q \leqslant n + k - 1 + q_2 + n - k -1 + q_3 + s - 1 + n = 4n - k + 2q_3 - 3 < 4n + 2q_3  \; .
                \label{formula_q_3_bound_in_theorem_L_n_q_bound_for_arbitrary_q}
            \end{equation}
            Следовательно, $\log_2 q_3 > \log_2 (q - 4n) - 1$. Отсюда получаем соотношение
            $$
                L(f,q) \lesssim 2^n + \frac{32n2^n}{\log_2 (q - 4n) - \log_2 n - 2} \; ,
            $$
            которое верно при $\log_2 (q - 4n) > \log_2 n + 2 \Rightarrow q > 8n$.
            \smallskip

        \item
            Пусть $t_2 \leqslant pn2^{n-k} \sim n2^n \mathop / s$, $t_3 \leqslant 2^s$. В этом случае 
            \begin{multline}
                L(f,q) \lesssim 2n + \frac{n2^n}{s} + q_1 + q_2 + \frac{4q_3 2^k}{s} +
                     \frac{8k2^{k+1}}{\log_2 q_1 - \log_2 k - 1} + \\
                        + \frac{8n2^n}{\log_2 q_2 - \log_2 (n-k) - 1}
                        + \frac{32 \cdot 2^n}{\log_2 q_3 - \log_2 s - 1} \; .
                    \label{formula_L_f_q_second_way}
            \end{multline}
            
            Как и в первом способе, положим $s = n - k$, $k = \lceil n \mathop / \phi(n) \rceil$,
            где $\phi(n) \leqslant n \mathop / (\log_2 n + \log_2 \psi(n))$
            и $\psi(n)$~--- любые сколь угодно медленно растущие функции;
            $q_1 = 0$, $q_2 = q_3$.
            Тогда из рассуждений, приведенных при описании первого способа, следует соотношение
            $$
                L(f,q) \lesssim 2^n + \frac{8n2^n}{\log_2 q_2 - \log_2 (n-k) - 1}
                        + \frac{32 \cdot 2^n}{\log_2 q_3 - \log_2 s - 1} \; .
            $$

            Очевидно, что
            $32 \cdot 2^n \mathop / (\log_2 q_2 - \log_2 (n-k) - 1) = o \left(8n2^n \mathop / (\log_2 q_2 - \log_2 (n-k) - 1) \right)$,
            поэтому верно соотношение 
            $$
                L(f,q) \lesssim 2^n + \frac{8n2^n}{\log_2 q_2 - \log_2 (n-k) - 1}  \; .
            $$
            
            Поскольку $\log_2 q_3 > \log_2 (q - 4n) - 1$, то и $\log_2 q_2 > \log_2(q - 4n) - 1$.
            Отсюда получаем соотношение
            $$
                L(f,q) \lesssim 2^n + \frac{8n2^n}{\log_2 (q - 4n) - \log_2 n - 2} \; ,
            $$
            которое верно при $\log_2 (q - 4n) > \log_2 n + 2 \Rightarrow q > 8n$.

            Видно, что второй способ синтеза асимптотически лучше первого.
            \smallskip
    \end{enumerate}

    Поскольку мы описали алгоритм синтеза обратимой схемы для произвольного отображения $f$, то
    $$
        L(n,q) \leqslant L(f,q) \lesssim 2^n + \frac{8n2^n}{\log_2 (q-4n) - \log_2 n - 2} \;
    $$
    при $q > 8n$.
    При значениях $q \gtrsim n 2^{n-\lceil n \mathop / \phi(n)\rceil}$
    не наблюдается снижение сложности обратимой схемы для рассматриваемого способа синтеза.
\end{proof}

\section{Зависимость глубины обратимой схемы от количества дополнительных входов}

Из доказательства Теоремы~\ref{theorem_L_n_q_bound_for_arbitrary_q} также можно получить верхнюю оценку для функции $D(n,q)$
в случае $q > 8n$, $q \lesssim n 2^{n-o(n)}$, но для этого необходимо сперва доказать вспомогательную лемму.

\begin{lemma}\label{lemma_D_CONJ_bound}
    Для любого значения $q > 2n$, $q \lesssim 2^n$ верны соотношения
    \begin{align*}
        D_{CONJ}(n, q, t) &\leqslant q + 2t(2+\log_2 n - \log_2 (\log_2 q - \log_2 n - 1)) \; . \\
        D_{CONJ}(n, 0, t) &\leqslant 2t \cdot \lceil \log_2 n \rceil  \; .
    \end{align*}
\end{lemma}
\begin{proof}
    Рассмотрим схему $\frS_{CONJ(n,q)}$ из Леммы~\ref{lemma_L_CONJ_bound}.
    Согласно формуле~\eqref{formula_bound_for_L_CONJ_with_r}, верно неравенство
    $L_{CONJ}(n, q, t) \leqslant q + 2t \cdot 2^r$. Промежуточные значения, хранимые на $q$ дополнительных входах,
    можно получить с глубиной не более $q$.
    Также очевидно, что сконструировать по мере необходимости один значимый выход схемы $\frS_{CONJ}$ на первых $r$ уровнях
    можно с глубиной $r$, см. рис.~\ref{pic_construct_conjunctions_on_demand}.
    Отсюда следует, что
    $$
        D_{CONJ}(n, q, t) \leqslant q + 2tr  \; .
    $$

    Согласно формуле~\eqref{formula_for_2_pow_r_in_lemma_L_CONJ_bound}, при $q > 2n$ верно неравенство
    $$
        2^r < \frac{4n}{\log_2 q - \log_2 n - 1}  \; ,
    $$
    откуда следует, что
    \begin{gather*}
        r < 2 + \log_2 n - \log_2 (\log_2 q - \log_n - 1)  \; , \\
        D_{CONJ}(n, q, t) \leqslant q + 2t(2+\log_2 n - \log_2 (\log_2 q - \log_2 n - 1)) \; .
    \end{gather*}
    
    Соотношение $D_{CONJ}(n, 0, t) \leqslant 2t \cdot \lceil \log_2 n \rceil$ следует из соотношения~\eqref{formula_L_CONJ_0}
    и того факта, что сконструировать одну конъюнкцию $x_1^{a_1} \wedge \ldots \wedge x_n^{a_n}$ можно с логарифмической глубиной
    $\lceil \log_2 n \rceil$.
\end{proof}

Аналогично, для обратимой схемы $\frS_{XOR(n,q)}$ верно неравенство
\begin{multline*}
    D_{XOR}(n, q, t) = 2D_{CONJ}(n, q, t) \leqslant 2q + 4t(2+\log_2 n - \log_2 (\log_2 q - \log_2 n - 1))
\end{multline*}
для любого значения $q > 2n$, $q \lesssim 2^n$.

Итак, докажем последнюю теорему данной работы.
\begin{theorem}[общая верхняя оценка глубины обратимой схемы с дополнительными входами]\label{theorem_D_n_q_bound_for_arbitrary_q}
    Для любого значения $q > 8n$, $q \lesssim n 2^{n-\lceil n \mathop / \phi(n)\rceil}$,
    где $\phi(n) \leqslant n \mathop / (\log_2 n + \log_2 \psi(n))$
    и $\psi(n)$~--- любые сколь угодно медленно растущие функции, верно соотношение
    $$
        D(n,q) \lesssim 2^{n+1}(2,5 + \log_2 n - \log_2 (\log_2 (q - 4n) - \log_2 n - 2))  \; .
    $$
\end{theorem}
\begin{proof}
    Рассмотрим обратимую схему $\frS_f$ из доказательства Теоремы~\ref{theorem_L_n_q_bound_for_arbitrary_q},
    синтезированную алгоритмом $\mathbf {A_q}$.
    Из соотношения~\eqref{formula_L_f_q_common} можно вывести аналогичное соотношение для глубины $D(f,q)$ вида
    \begin{multline*}
        D(f,q) = 2 + D_{CONJ}(k, q_1, t_1) + D_{CONJ}(n-k, q_2, t_2) + \\
            + 2p \cdot D_{XOR}(s, q_3, t_3) + pn2^{n-k}  \; .
    \end{multline*}

    Положим $s = n - k$, $k = \lceil n \mathop / \phi(n) \rceil$, где $\phi(n) \leqslant n \mathop / (\log_2 n + \log_2 \psi(n))$
    и $\psi(n)$~--- любые сколь угодно медленно растущие функции.
    В этом случае будут верны соотношения $p \sim 2^k \mathop / s \geqslant \psi(n)$
    и $pn2^{n-k} \sim n2^n \mathop / s \sim 2^n$.
    
    Положим $q_1 = 0$. Поскольку $t_1 = 2^{k+1}$, то
    $$
        D_{CONJ}(k, q_1, t_1) \leqslant 2t_1 \cdot \lceil \log_2 k \rceil
            \leqslant \lceil \log_2 \lceil n \mathop / \phi(n)\rceil \rceil \cdot 2^{\lceil n \mathop / \phi(n) \rceil + 2} = o(2^n)  \; .
    $$
    Таким образом, верно соотношение
    $$
        D(f,q) \lesssim 2^n + D_{CONJ}(n-k, q_2, t_2) + \frac{2^{k+1}}{s} \cdot D_{XOR}(s, q_3, t_3)  \; .
    $$
   
    Рассмотрим те же два случая для $t_2$ и $t_3$, что и на с.~\pageref{item_first_case_for_theorem_L_n_q_bound_for_arbitrary_q}.
    \begin{enumerate}
        \item
            Пусть $t_2 = p2^{n-k} \sim 2^n \mathop / s$, $t_3 \leqslant n2^{n-k}$.
            В этом случае
            \begin{gather*}
                D_{CONJ}(n-k, q_2, t_2) \leqslant q_2 + \frac{2^{n+1}}{s}(2+\log_2 s - \log_2 (\log_2 q_2 - \log_2 s - 1))  \; , \\
                D_{XOR}(s, q_3, t_3) \leqslant 2 q_3 + n2^{n-k+2}(2+\log_2 s - \log_2 (\log_2 q_3 - \log_2 s - 1))  \; .
            \end{gather*}
            
            Положим $q_2 = q_3$. Обозначим $d = 2+\log_2 s - \log_2 (\log_2 q_2 - \log_2 s - 1)$, тогда
            $$
                D(f,q) \lesssim 2^n + q_2 + \frac{d2^{n+1}}{s} +
                    \frac{q_3 2^{k+2}}{s} + \frac{dn2^{n+3}}{s}  \; .
            $$
            
            Согласно формуле~\eqref{formula_least_member_bound_first_case_in_theorem_L_n_q_bound_for_arbitrary_q},
            верно соотношение
            $$
                2^n + q_2 + \frac{4q_3 2^k}{s} \lesssim 2^n  \; .
            $$
            Отсюда получаем, что
            $$
                D(f,q) \lesssim 2^n + \frac{dn2^{n+3}}{s} \lesssim
                    2^n + 2^{n+3}(2+\log_2 n - \log_2 (\log_2 q_3 - \log_2 n - 1))  \; .
            $$

            Из соотношения~\eqref{formula_q_3_bound_in_theorem_L_n_q_bound_for_arbitrary_q} следует, что
            $\log_2 q_3 > \log_2 (q - 4n) - 1$. Таким образом, получаем итоговую оценку сверху вида
            $$
                D(f,q) \lesssim 2^n(17 + 8(\log_2 n - \log_2 (\log_2 (q - 4n) - \log_2 n - 2)))  \; ,
            $$
            которая верна при $\log_2 (q - 4n) > \log_2 n + 2 \Rightarrow q > 8n$.
            \smallskip

        \item
            Пусть $t_2 \leqslant pn2^{n-k} \sim n2^n \mathop / s$, $t_3 \leqslant 2^s$. В этом случае
            \begin{gather*}
                D_{CONJ}(n-k, q_2, t_2) \leqslant q_2 + \frac{n2^{n+1}}{s}(2+\log_2 s - \log_2 (\log_2 q_2 - \log_2 s - 1))  \; , \\
                D_{XOR}(s, q_3, t_3) \leqslant 2 q_3 + 2^{s+2}(2+\log_2 s - \log_2 (\log_2 q_3 - \log_2 s - 1))  \; .                    
            \end{gather*}
            
            Положим $q_2 = q_3$. Обозначим $d = 2+\log_2 s - \log_2 (\log_2 q_2 - \log_2 s - 1)$, тогда
            $$
                D(f,q) \lesssim 2^n + q_2 + \frac{dn2^{n+1}}{s} + \frac{q_3 2^{k+2}}{s} + \frac{d2^{n+3}}{s}  \; .
            $$
            
            Согласно формуле~\eqref{formula_least_member_bound_first_case_in_theorem_L_n_q_bound_for_arbitrary_q},
            верно соотношение
            $$
                2^n + q_2 + \frac{4q_3 2^k}{s} \lesssim 2^n  \; .
            $$
            Отсюда получаем, что
            $$
                D(f,q) \lesssim 2^n + \frac{dn2^{n+1}}{s} \lesssim
                    2^n + 2^{n+1}(2+\log_2 n - \log_2 (\log_2 q_2 - \log_2 n - 1))  \; .
            $$

            Из соотношения~\eqref{formula_q_3_bound_in_theorem_L_n_q_bound_for_arbitrary_q} следует, что
            $\log_2 q_2 > \log_2 (q - 4n) - 1$. Таким образом, получаем итоговую оценку сверху вида
            $$
                D(f,q) \lesssim 2^{n+1}(2,5 + \log_2 n - \log_2 (\log_2 (q - 4n) - \log_2 n - 2))  \; ,
            $$
            которая верна при $\log_2 (q - 4n) > \log_2 n + 2 \Rightarrow q > 8n$.

            Видно, что второй способ синтеза асимптотически лучше первого.
            \smallskip
    \end{enumerate}
    
    Поскольку мы описали алгоритм синтеза обратимой схемы для произвольного отображения $f$, то
    $$
        D(n,q) \leqslant D(f,q) \lesssim 2^{n+1}(2,5 + \log_2 n - \log_2 (\log_2 (q - 4n) - \log_2 n - 2)) \;
    $$
    при $q > 8n$.
\end{proof}

При увеличении количества дополнительных входов с $q \sim n 2^{n-o(n)}$ до $q \sim 2^n$
верхняя асимптотическая оценка функции $D(n,q)$ снижается с экспоненциальной до линейной,
см. соотношение~\eqref{formula_linear_depth_bound}.
Однако выведение зависимости верхней оценки функции $D(n,q)$ от $q$ для диапазона $n 2^{n-o(n)} \lesssim q \lesssim 2^n$
выходит за рамки данной работы.


\section*{Заключение}

В данной работе были рассмотрены обратимые схемы, состоящие из функциональных элементов NOT, CNOT и 2-CNOT
и имеющие различное число дополнительных входов $q$.
Были изучены функции Шеннона сложности $L(n, q)$ и глубины $D(n,q)$ обратимой схемы,
реализующей какое-либо отображение $\ZZ_2^n \to \ZZ_2^n$ в условиях ограничения на значение $q$.
Были доказаны верхние оценки для функций $L(n, q)$ и $D(n, q)$ для диапазона $8n < q \lesssim n2^{n-o(n)}$.
Из полученных соотношений можно сделать вывод, что использование дополнительной памяти в обратимых схемах,
состоящих из элементов NOT, CNOT и 2-CNOT, почти всегда позволяет существенно снизить сложность и глубину таких схем.


\end{document}